\documentclass[journal,draftcls,onecolumn,11pt]{IEEEtran}
\normalsize

\usepackage{amsmath,amssymb,amsthm}
\usepackage{enumerate}
\usepackage{stfloats}
\usepackage{comment}
\usepackage{subcaption}
\usepackage{siunitx}
\usepackage{cite}

\usepackage[dvipsnames]{xcolor}
\definecolor{myPink}{RGB}{255,105,183}

\usepackage[T1]{fontenc}
\usepackage{graphics} 
\usepackage{epsfig} 
\usepackage[mathscr]{euscript}
\usepackage{algorithm}
\usepackage[noend]{algpseudocode}
\makeatletter
\def\BState{\State\hskip-\ALG@thistlm}
\makeatother

\usepackage{soul}
\usepackage{tikz}
\usetikzlibrary{arrows,shapes,chains,matrix,positioning,scopes,patterns,calc,
decorations.markings,
decorations.pathmorphing,
}

\usepackage{pgfplots}
\pgfplotsset{compat=1.3}
\usepgflibrary{shapes}

\newtheorem{theorem}{Theorem}
\newtheorem{lemma}[theorem]{Lemma}

\newtheorem{definition}[theorem]{Definition}
\newtheorem{example}[theorem]{Example}

\renewcommand{\epsilon}{\varepsilon}

\newcommand{\mb}[1]{\mathbf{#1}}

\newcommand{\RNum}[1]{\uppercase\expandafter{\romannumeral #1\relax}}

\def\Pr{\mathrm{Pr}}

\DeclareMathAlphabet{\mcl}{OMS}{cmsy}{m}{n}

\newlength\tikzwidth
\newlength\tikzheight

\definecolor{mycolor1}{rgb}{0.63529,0.07843,0.18431}%
\definecolor{mycolor2}{rgb}{0.00000,0.44706,0.74118}%
\definecolor{mycolor3}{rgb}{0.00000,0.49804,0.00000}%
\definecolor{mycolor4}{rgb}{0.87059,0.49020,0.00000}%
\definecolor{mycolor5}{rgb}{0.00000,0.44700,0.74100}%
\definecolor{mycolor6}{rgb}{0.74902,0.00000,0.74902}%

\usepackage{amsthm}
\usepackage{amsmath,amsfonts}
\usepackage{cite}
\usepackage{amssymb}
\usepackage{dsfont}
\usepackage{graphicx}
\usepackage{color}
\usepackage{mathtools}
\usepackage{bbm}
\usepackage{latexsym}

\newif\iflonger
\longertrue
\newif\ifhashcoll
\hashcollfalse
\newif\ifproof
\prooftrue

\usepackage{pgfplots}
\usepackage{tikz}
\def\fig_path{./Figures}

\usepackage{bbm}

\newcommand{\myfloor}[1]{\lfloor #1 \rfloor}
\newcommand{\myceil}[1]{\lceil #1 \rceil}

\begin{document}
\title {{Random Khatri-Rao-Product Codes for Numerically-Stable Distributed Matrix Multiplication}}

%
\author{\begin{tabular}{c} Adarsh M. Subramaniam,
Anoosheh Heidarzadeh,
Krishna R. Narayanan
\end{tabular} \\
Department of Electrical and Computer Engineering \\
Texas A\&M University
}
\maketitle

\begin{abstract}
We propose a class of codes called random Khatri-Rao-Product (RKRP) codes for distributed matrix multiplication in the presence of stragglers.
The main advantage of the proposed codes is that decoding of RKRP codes is highly numerically stable in comparison to decoding of Polynomial codes \cite{yu2018polycode} and decoding of the recently proposed OrthoPoly codes \cite{fahim2019}.
We show that RKRP codes are maximum distance separable with probability 1.
The communication cost and encoding complexity for RKRP codes are identical to that of OrthoPoly codes and Polynomial codes and
the average decoding complexity of RKRP codes is lower than that of OrthoPoly codes.
Numerical results show that the average relative $L_2$-norm of the reconstruction error for RKRP codes is substantially better than that of OrthoPoly codes.
\end{abstract}


\section{Introduction and Main Results}
We consider the problem of computing $\mb{A}^{\mathsf{T}}\mb{B}$ for two matrices $\mathbf{A} \in \mathbb{R}^{N_2 \times N_1}$ and $\mathbf{B} \in \mathbb{R}^{N_2 \times N_3}$ in a distributed fashion  using a coded matrix multiplication scheme with $N$ worker nodes \cite{yu2018polycode,yu2018straggler,dutta2018unified,unidecode,lagrange_coded,speeding_up,shortdot,coded_unified,highdim,stserv_mult,inv_pro}.
In \cite{yu2018polycode}, Yu, Maddah-Ali and Avestimehr proposed an elegant encoding scheme called Polynomial codes in which the matrices $\mb{A}^{\mathsf{T}}$ and $\mb{B}$ are each split into $m$ and $n$ sub-matrices, respectively, the sub-matrices are encoded using a polynomial code and the computations are distributed to $N$ worker nodes.
This scheme is shown to have optimal recovery threshold, i.e., the matrix product $\mb{A}^{\mathsf{T}}\mb{B}$ can be computed (recovered) using the results of computation from any subset of worker nodes of cardinality $K=mn$.
In the language of coding theory, Polynomial codes are generalized Reed-Solomon codes, their generator matrices have Vandermonde structures, and they are
maximum distance separable (MDS) codes.

One important drawback of Polynomial codes is that the process of recovering $\mb{A}^{\mathsf{T}}\mb{B}$ from the results of the worker nodes (the decoding process) involves explicitly or implicitly inverting a Vandermonde matrix, which is well known to be highly numerically unstable even for moderate values of $K = mn$.
Very recently, Fahim and Cadambe \cite{fahim2019} proposed a very interesting polynomial code called OrthoPoly code which uses an orthogonal polynomial basis resulting in a Chebyshev-Vandermonde structure for the generator matrix.
OrthoPoly codes are also MDS codes, i.e., have optimal recovery threshold; however, they afford better numerical stability than Polynomial codes.
In particular, when there are $S$ stragglers among $N$ nodes, i.e., $N = K+S$, the condition number of the matrix that needs to be inverted grows only polynomially in $N$.
However, the main drawback of OrthoPoly codes is that the condition number still grows exponentially in $S$ making it unsuitable even for moderately large values of $S$.

In this paper, we propose a coding scheme for the distributed matrix multiplication problem which we call Random Khatri-Rao-Product (RKRP) codes which exhibits substantially better numerical stability than Polynomial codes \cite{yu2018polycode} and OrthoPoly codes~\cite{fahim2019}. The proposed coding scheme is not based on polynomial interpolation; rather, it is designed in the spirit of random codes in information theory.

RKRP codes split both $\mb{A}^{\mathsf{T}}$ and $\mb{B}$ into sub-matrices and encode them by forming random linear combinations of the sub-matrices.
The proposed RKRP codes have several desirable features: (i) RKRP codes have the same thresholds, encoding complexity and communication cost as that of Polynomial codes and OrthoPoly codes; (ii) Decoding process of RKRP codes is substantially more numerically stable than that of Polynomial and OrthoPoly codes, and decoding can be implemented even for fairly large values of $K,S,N$ (e.g., $K=1000$ and any $S,N$); and (iii) decoding complexity of RKRP codes is lower than that of OrthoPoly codes.
To the best of our knowledge, the RKRP code construction and the analysis of their
MDS property are new.

We present two ensembles of generator matrices for RKRP codes called the non-systematic RKRP ensemble and the systematic RKRP ensemble. Codes from these ensembles will be referred to
as \emph{non-systematic RKRP codes} and \emph{systematic RKRP codes}\footnote{The terminology of associating the words systematic and non-systematic with the code, rather than with the encoder is not standard in coding theory. While it is possible to find a systematic encoder for a non-systematic RKRP code, the resulting code would not belong to the systematic RKRP ensemble and hence, should be treated as a non-systematic RKRP code.}, respectively.
Systematic RKRP codes have better average decoding complexity and better numerical stability.
Hence, systematic RKRP codes would be preferred over non-systematic RKRP codes for most applications.
However, we present both non-systematic and systematic ensembles in this paper for the following reasons.
Since Polynomial and OrthoPoly codes are presented with non-systematic encoding, non-systematic RKRP codes allow for a fair comparison with Polynomial and OrthoPoly codes.
The proofs are also easier to follow when presented for the non-systematic ensemble first and then extended to the systematic ensemble.
Finally, non-systematically RKRP codes provide privacy which systematic RKRP codes do not, although this issue is not studied further in this paper.

\section{Notation}
We use boldface capital letters for matrices and underlined variables to represent vectors. We denote the $i,j$th element of the matrix $\mathbf{A}$ by $[\mathbf{A}]_{i,j}$. The $i$th row and $i$th column of matrix $\mathbf{A}$ will be represented by $[\mathbf{A}]_{i,:}$ and $[\mathbf{A}]_{:,i}$, respectively.
If $\mathcal{S}_1 \subset \mathbb{Z}^+$ and $\mathcal{S}_2 \subset \mathbb{Z}^+$ are two subsets of positive integers, then the submatrix of $\mathbf{A}$ corresponding to the rows from $\mathcal{S}_1$ and columns from $\mathcal{S}_2$ is given by $[\mathbf{A}]_{\mathcal{S}_1,\mathcal{S}_2}$.
We denote the set of integers from $i$ to $j$, inclusive of $i$ and $j$ by $i:j$ and we denote the set of integers from 1 to $i$ by $[i]$.
For a vector $\underline{v}$, we denote the part of vector $\underline{v}$ between indices $i$ and $j$ as $\underline{v}_{i:j}$. We will assume that vectors without transposes are column vectors unless stated otherwise.
Random variables will be denoted by capital letters and their realizations will be denoted by lower case letters.

\section{System Model and Preliminaries} \label{sm&p}
We consider a system with one master node which has access to matrices $\mb{A}^{\mathsf{T}}$ and $\mathbf{B}$ and $N$ worker nodes which can perform multiplication of sub-matrices of $\mb{A}^{\mathsf{T}}$ and $\mathbf{B}$.
At the master node, the matrix $\mb{A}^{\mathsf{T}}$ is split into $m$ sub-matrices row-wise and $\mathbf{B}$ is split into $n$ sub-matrices column-wise as shown below.
\begin{equation}
    \mathbf{A}^{\mathsf{T}} = \begin{bmatrix}
\mathbf{A}_1^{\mathsf{T}} \\
\mathbf{A}_2^{\mathsf{T}}  \\
\vdots \\
\mathbf{A}_m^{\mathsf{T}}  \\
\end{bmatrix}, \quad
\mathbf{B} = \begin{bmatrix} \mathbf{B}_1 & \mathbf{B}_2 & \cdots & \mathbf{B}_n \end{bmatrix}
\end{equation}

In order to compute the matrix product $\mathbf{A}^{\mathsf{T}} \mathbf{B}$, we need to compute the matrix products
$\mathbf{A}_j^{\mathsf{T}} \mathbf{B}_l$ for $j=1,\ldots,m$ and $l=1,\ldots,n$.
The main idea in distributed coded computation is to first encode $\mathbf{A}_1^{\mathsf{T}}, \ldots, \mathbf{A}_m^{\mathsf{T}}$ and $\mathbf{B}_1, \ldots, \mathbf{B}_n$ into $N$ pairs of matrices $\big(\mathbf{U}_i^{\mathsf{T}},\mathbf{V}_i\big), i=1, \ldots, N$.\footnote{This is not the most general form of encoding but many of the existing encoding schemes in the literature as well as the proposed
scheme can be represented in this way.}
The $i$th worker node is then tasked with computing the matrix product $\mathbf{X}_i = \mathbf{U}_i^{\mathsf{T}} \mathbf{V}_i$.
It is assumed that $K$ out of the $N$ workers return the result of their computation; these worker nodes are called non-stragglers.
Without loss of generality we assume that the non-stragglers are worker nodes $1,\ldots, K$.

\begin{definition}
An \emph{encoding scheme} is a mapping from $(\mathbf{A}_1^{\mathsf{T}}, \ldots, \mathbf{A}_m^{\mathsf{T}}, \mathbf{B}_1, \ldots, \mathbf{B}_n)$ to
$\{ (\mb{X}_i = \mathbf{U}_i^{\mathsf{T}}, \mathbf{V}_i ) \}$ for $i=1,\ldots, N$. A \emph{codeword} is a vector of matrices $\underline{\mb{X}}=[\mb{X}_1, \mb{X}_2, \ldots, \mb{X}_N]$. A \emph{code} is the set of possible codewords $\{\underline{\mb{X}}\}$.
\end{definition}

\begin{definition}
An encoding scheme is said to result in a maximum distance separable (MDS) code,
or the corresponding code is said to be MDS, if the set of matrix products $\{\mb{A}_i^{\mathsf{T}}\mb{B}_j\}$ for $i=1,\ldots,m$ and $j=1,\ldots,n$ can be computed (recovered) from any subset of $\{\mathbf{X}_1, \mathbf{X}_2, \ldots, \mathbf{X}_N\}$ of size $mn$, where $\mathbf{X}_i = \mathbf{U}_i^{\mathsf{T}} \mathbf{V}_i$.
\end{definition}

\begin{definition}
The row-wise Khatri-Rao product of two matrices $\mathbf{P} \in \mathbb{R}^{K \times m}$ and
$\mathbf{Q} \in \mathbb{R}^{K \times n}$ denoted by $\mathbf{P} \odot \mathbf{Q}$
is given by the matrix $\mathbf{M}$ whose $i$th row is the Kronecker product of the $i$th row of $\mathbf{P}$
and the $i$th row of $\mathbf{Q}$, i.e.,
\begin{equation}
    [\mathbf{M}]_{i,:} = [\mathbf{P}_{i,:}] \otimes [\mathbf{Q}_{i,:}]
\end{equation}
where $\otimes$ refers to the Kronecker product.
\end{definition}

\section{Non-Systematically encoded Random Khatri-Rao-Product Codes}
\subsection{Encoding:}
Our proposed non-systematic RKRP codes are encoded as follows. For  $i=1,\ldots, N$, the master node computes
\begin{eqnarray}
    \mathbf{U}_i^{\mathsf{T}}  & = & \sum_{j=1}^m p_{i,j} \mathbf{A}_j^{\mathsf{T}}, \\
    \mathbf{V}_i & = & \sum_{l=1}^n q_{i,l} \mathbf{B}_l
\label{eqn:encoding}
\end{eqnarray}
where $p_{i,j}, q_{i,l}$ are realizations of independent identically distributed random variables $P_{i,j}$ and $Q_{i,l}$, respectively.
Both $P_{i,j}$ and $Q_{i,l}$ are assumed to be continuous random variables with a probability density function $f \  \forall i,j,l$, i.e., their distribution is absolutely continuous with respect to the Lebesgue measure.
$\mathbf{U}_i$ and  $\mathbf{V}_i$ are then transmitted to the $i$th worker node which is tasked with computing
$\mathbf{X}_i = \mathbf{U}_i^{\mathsf{T}} \mathbf{V}_i$.
We first note that $\mathbf{X}_i$ can be written as
\begin{align}\nonumber
    \mathbf{X}_i & = \bigg(\sum_{j=1}^m p_{i,j} \mathbf{A}_j^{\mathsf{T}} \bigg) \bigg( \sum_{l=1}^n q_{i,l} \mathbf{B}_l \bigg)\\
    & = \sum_{j=1}^m \sum_{l=1}^n p_{i,j} q_{i,l} \mathbf{A}_j^{\mathsf{T}} \mathbf{B}_l.
\end{align}

Since the matrix $\mathbf{X}_i$ is a linear combination of the desired matrix products $\mathbf{A}_j^{\mathsf{T}} \mathbf{B}_l$, the $(s,t)$th entry of $\mathbf{X}_i$, namely $[\mathbf{X}_i]_{s,t}$, is a linear combination of the $(s,t)$th entries of the matrix products $\mathbf{A}_j^{\mathsf{T}} \mathbf{B}_l$, namely $[\mathbf{A}_j^{\mathsf{T}} \mathbf{B}_l]_{s,t}$. During the decoding process, we attempt to recover $[\mathbf{A}_j^{\mathsf{T}} \mathbf{B}_l]_{s,t}$ from $[\mathbf{X}_1]_{st}, \ldots, [\mathbf{X}_K]_{s,t}$ for each pair of $s,t$ separately.

To keep the discussions clear, we focus on the recovery of the $(1,1)$th entry of $\mathbf{A}_j^{\mathsf{T}} \mathbf{B}_l$, namely $[\mathbf{A}_j^{\mathsf{T}} \mathbf{B}_l]_{1,1}$.
The same idea extends to the recovery of other indices as well.
Let $y_i = [\mb{X}_i]_{1,1}$ denote the $(1,1)$th entry in the matrix product computed by the $i$th non-straggler worker node, and let $z_{j,l} = [\mathbf{A}_j^{\mathsf{T}} \mathbf{B}_l]_{1,1}$.

The vector of computed values can be written as a linear combination of matrix products given by
\begin{equation}
\label{eqn:matrix1}
    \begin{bmatrix}
    y_1 \\
    y_2 \\
    \vdots\\
    y_i \\
    \vdots \\
    y_N \\
    \end{bmatrix}
    =
    \begin{bmatrix}
    p_{1,1}q_{1,1} & p_{1,1}q_{1,2} & \ldots & p_{1,1}q_{1,n} & \ldots & p_{1,m}q_{1,n}\\
    p_{2,1}q_{2,1} & p_{2,1}q_{2,2} & \ldots & p_{2,1}q_{2,n} & \ldots & p_{2,m}q_{2,n}\\
    & & \vdots & & \\
    p_{i,1}q_{i,1} & p_{i,1}q_{i,2} & \ldots & p_{i,1}q_{i,n} & \ldots & p_{i,m}q_{i,n}\\
    & & \vdots & & \\
    p_{N,1}q_{N,1} & p_{N,1}q_{N,2} & \ldots & p_{N,1}q_{N,n} & \ldots & p_{N,m}q_{N,n}\\
    \end{bmatrix}
    \begin{bmatrix}
    z_{1,1} \\
    z_{1,2} \\
    \vdots\\
    z_{1,n} \\
    \vdots \\
    z_{m,n} \\
    \end{bmatrix}
\end{equation}
It will be more convenient to express \eqref{eqn:matrix1} in a slightly different form.
For $j \in \{1,\ldots,mn\}$, let $j' = \lceil j/n \rceil$ and $j'' = (j-1) \mod n + 1$, and let
$w_j = z_{j',j''}$.
Without loss of generality, let us assume that the worker nodes which return their computation are
worker nodes $1,2, \dots, K$.
The computed values $y_i$'s are related to the unknown values $w_j$'s according to
\begin{equation}
\label{eqn:matrix2}
    \begin{bmatrix}
    y_1 \\
    y_2 \\
    \vdots\\
    y_i \\
    \vdots \\
    y_K \\
    \end{bmatrix}
    =
    \begin{bmatrix}
    p_{1,1}q_{1,1} & p_{1,1}q_{1,2} & \ldots & p_{1,1}q_{1,n} & \ldots & p_{1,m}q_{1,n}\\
    p_{2,1}q_{2,1} & p_{2,1}q_{2,2} & \ldots & p_{2,1}q_{2,n} & \ldots & p_{2,m}q_{2,n}\\
    & & \vdots & & \\
    p_{i,1}q_{i,1} & p_{i,1}q_{i,2} & \ldots & p_{i,1}q_{i,n} & \ldots & p_{i,m}q_{i,n}\\
    & & \vdots & & \\
    p_{K,1}q_{K,1} & p_{K,1}q_{K,2} & \ldots & p_{K,1}q_{K,n} & \ldots & p_{K,m}q_{K,n}\\
    \end{bmatrix}
    \begin{bmatrix}
    w_1 \\
    w_{2} \\
    \vdots\\
    w_{j} \\
    \vdots \\
    w_{K} \\
    \end{bmatrix}
\end{equation}
or, more succinctly as
\begin{equation}\label{eqn:mainequation}
  \underline{y} = \mathbf{G} \ \underline{w}
\end{equation}
where $\underline{y} = [y_1, y_2, \ldots, y_K]^{\mathsf{T}}$,
$\underline{w} = [w_1, w_2, \ldots, w_{mn}]^{\mathsf{T}}$, and
$\mathbf{G}$ is an $N \times mn$ generator matrix for a code
with $[\mathbf{G}]_{i,j} = p_{i,j'} q_{i,j''}$.

Let $\mathbf{P}$ and $\mathbf{Q}$ be two matrices  whose entries are given by
$[\mathbf{P}]_{i,j'} = p_{i,j'}$ and $[\mathbf{Q}]_{i,j''} = q_{i,j''}$.
It can be seen that
\begin{equation}\label{eqn:krprod}
\mathbf{G}  = \mathbf{P} \odot \mathbf{Q},
\end{equation} i.e., $\mathbf{G}$ is the row-wise Khatri-Rao product of two matrices
$\mathbf{P}$ and $\mathbf{Q}$.
Hence, we call these codes as Random Khatri-Rao-Product codes.

\begin{example}\label{ex:example1}
In order to clarify the main idea, consider an example with $m=2$ and $n=3$ and $N>6$. Without loss of generality, assume that the worker nodes $1,2,\ldots,6$ return the results of their computations, namely, $\mathbf{X}_1, \ldots, \mathbf{X}_6$. In this case, the set of computations returned by the worker nodes is related to the matrix products that we need to compute according to
\begin{equation}
    \begin{bmatrix}
    y_1 \\
    y_2 \\
    y_3 \\
    y_4 \\
    y_5 \\
    y_6 \\
    \end{bmatrix}
    =
    \begin{bmatrix}
    p_{1,1}q_{1,1} & p_{1,1}q_{1,2} & p_{1,1}q_{1,3} & p_{1,2}q_{1,1} & p_{1,2}q_{1,2} & p_{1,2}q_{1,3}\\
    p_{2,1}q_{2,1} & p_{2,1}q_{2,2} & p_{2,1}q_{2,3} & p_{2,2}q_{2,1} & p_{2,2}q_{2,2} & p_{2,2}q_{2,3}\\
    p_{3,1}q_{3,1} & p_{3,1}q_{3,2} & p_{3,1}q_{3,3} & p_{3,2}q_{3,1} & p_{3,2}q_{3,2} & p_{3,2}q_{3,3}\\
    p_{4,1}q_{4,1} & p_{4,1}q_{4,2} & p_{4,1}q_{4,3} & p_{4,2}q_{4,1} & p_{4,2}q_{4,2} & p_{4,2}q_{4,3}\\
    p_{5,1}q_{5,1} & p_{5,1}q_{5,2} & p_{5,1}q_{5,3} & p_{5,2}q_{5,1} & p_{5,2}q_{5,2} & p_{5,2}q_{5,3}\\
    p_{6,1}q_{6,1} & p_{6,1}q_{6,2} & p_{6,1}q_{6,3} & p_{6,2}q_{6,1} & p_{6,2}q_{6,2} & p_{6,2}q_{6,3}\\
    \end{bmatrix}
    \begin{bmatrix}
    w_{1} \\
    w_{2} \\
    w_{3} \\
    w_{4} \\
    w_{5} \\
    w_{6} \\
    \end{bmatrix}
\end{equation}
\end{example}

\begin{definition}
\label{def:nonsysensemble}
The ensemble of $N \times K$ generator matrices obtained by choosing the generator matrix $\mb{G}$ as in \eqref{eqn:krprod} where $p_{i,j}, q_{i,j}$ are realizations of random variables $P_{i,j}, Q_{i,j}$
such that $\{P_{1,1},\ldots,P_{N,m},Q_{1,1}, \ldots, Q_{N,n}\}$ is a set of independent random variables
with probability density function $f$
will be referred to as the \emph{non-systematic random Khatri-Rao-product generator matrix ensemble} $\mathcal{G}_{non-sys}(N,K,f)$.
\end{definition}

\subsection{Decoding:}
During decoding, an estimate of $\underline{w}$, namely $\underline{\hat{w}}$, is obtained as follows
\begin{equation}
    \label{eqn:nonsysdecoding}
    \underline{\hat{w}} = \mathbf{G}^{-1} \underline{y}.
\end{equation}

In the absence of numerical round-off errors, if $\mathbf{G}$ is invertible, then $\underline{\hat{w}} = \underline{w}$.
However, when performing computation with finite bits of precision,
there will be numerical errors in the computation. Let $\underline{e} = \underline{w}-\underline{\hat{w}}$ be the error,
and define the relative error as
\begin{equation}
    \label{eqn:fidelity}
    \eta := \frac{||\underline{e}||_2}{||\underline{w}||_2}.
\end{equation}

\section{Non-Systematic RKRP codes are MDS codes with probability 1}
\label{sec:nonsysMDS}
Our first main result in this paper is that if a generator matrix is randomly chosen from the non-systematic RKRP ensemble
$\mathcal{G}_{non-sys}(N,K,f)$, the encoding scheme defined in (\ref{eqn:encoding}) results in an MDS code with probability 1.

\begin{lemma}
\label{lem:new1}
Consider an analytic function $h(\underline{x})$ of several real variables $\underline{x}=[x_1,x_2,\cdots,x_n] \in \mathbbm{R}^n$ . If $h(\underline{x})$ is nontrivial
in the sense that there exists $\underline{x}_0 \in \mathbbm{R}^n$ such that $h(\underline{x}_0) \neq 0$ then the zero set of $h(\underline{x})$,
\[
\mathcal{Z}=\{\underline{x} \in \mathbbm{R}^n\  |\  h(\underline{x})=0\}
\]
is of measure (Lebesgue measure in $\mathbbm{R}^n$) zero.
\end{lemma}

\begin{proof}
This lemma is proved in~\cite[Lemma~1]{jiang2001almost} for the complex field $\mathbb{C}$.
The proof for the real field can be obtained by following the same steps and replacing $\mathbb{C}$ with $\mathbb{R}$.
\end{proof}

\begin{theorem} \label{th:1}
Non-systematic RKRP codes are MDS codes with probability 1.
\end{theorem}

\begin{proof}
To prove the theorem, we need to prove that the matrix $\mathbf{G}$ obtained when $K=mn$ in \eqref{eqn:matrix2} is a full rank matrix with probability 1. Let $p_{i,j}$ be a realization of the random variable $P_{i,j}$ and let $q_{i,j}$ be a realization of the random variable $Q_{i,j}$. The generator matrix in \eqref{eqn:matrix2} is a realization of the matrix of random variables $\{P_{i,j}\}_{i\in [K],j\in [m]}$ and $\{Q_{i,j}\}_{i\in [K],j\in [n]}$:
\begin{equation}
\mathbf{\Gamma} =
    \begin{bmatrix}
    P_{1,1}Q_{1,1} & P_{1,1}Q_{1,2} & \ldots & P_{1,1}Q_{1,n} & \ldots & P_{1m}Q_{1,n}\\
    P_{2,1}Q_{2,1} & P_{2,1}Q_{2,2} & \ldots & P_{2,1}Q_{2,n} & \ldots & P_{2m}Q_{2,n}\\
    & & \vdots & & \\
    P_{i,1}Q_{i,1} & P_{i,1}Q_{i2} & \ldots & P_{i,1}Q_{in} & \ldots & P_{i,m}Q_{in}\\
    & & \vdots & & \\
    P_{K,1}Q_{K,1} & P_{K,1}Q_{K,2} & \ldots & P_{K,1}Q_{K,n} & \ldots & P_{K,m}Q_{K,n}\\
    \end{bmatrix}
\end{equation}

We will show that $\Pr (\mathrm{rank}(\mathbf{\Gamma}) \neq mn ) = 0$. The determinant of $\mathbf{\Gamma}$ is a polynomial in the variables $\{P_{i,j}\}_{i\in [K],j\in [m]}$ and $\{Q_{i,j}\}_{i\in [K],j\in [n]}$ with degree $2mn$. Let
\begin{equation}
    \label{eqn:deteq}
\det({\mathbf{\Gamma}}) = h(P_{1,1},\ldots,P_{K,m},Q_{1,1},\ldots,Q_{K,n})
\end{equation}
We first show that there exists at least one $P_{1,1},\ldots,P_{K,m},Q_{1,1},\ldots,Q_{K,n}$ for which \[h(P_{1,1},\ldots,P_{K,m},Q_{1,1},\ldots,Q_{K,n}) \neq 0.\]
For $j \in [mn]$, let $j' = \lceil j/n \rceil$ and $j'' = ((j-1) \mod n) + 1$.
Let $P_{j,j'}=1, Q_{j,j''}=1, \forall j \in [mn]$, $P_{j,l} = 0, \forall j \in [mn], l \neq j'$, and $Q_{j,l} = 0,  \forall j \in [mn], l\neq j''$. For this choice of $P_{1,1},\ldots,P_{K,m},Q_{1,1},\ldots,Q_{K,n}$, it can be seen that the matrix $\mathbf{\Gamma}$ reduces to an identity matrix, and hence $h(P_{1,1},\ldots,P_{K,m},Q_{1,1},\ldots,Q_{K,n})=1$ ($\neq 0$).
From Lemma~\ref{lem:new1}, we then see that the zero set of $h$ has measure zero, and hence, $\Pr(\mathrm{rank}(\mathbf{\Gamma}) \neq mn) = 0$.

\end{proof}

\section{Systematic Khatri-Rao-Product Codes}
In this section, we introduce a systematic construction of random Khatri-Rao-Product codes which reduces the average encoding and decoding complexities compared to its non-systematic counterpart.

\subsection{Encoding:}
In systematic encoding, the first $K$ worker nodes are simply given the submatrices $\mathbf{A}_j^{\mathsf{T}}$ and $\mathbf{B}_l$ without encoding and the other $N-K$ worker nodes are given encoded versions as in the non-systematic version. For $i \in \{1,\ldots,K=mn\}$, let $i' = \lceil i/n \rceil$ and $i'' = ((i-1) \mod n) + 1$. The encoding process can be described as below
\begin{align}
\label{eqn:sysencoding}
&\mathbf{U}_i^{\mathsf{T}}=
\begin{cases}
\mathbf{A}_{i'}^{\mathsf{T}}, & i\in[K],\\
\sum_{j=1}^m p_{i-K,j} \mathbf{A}_j^{\mathsf{T}}, &  K+1\leq i \leq N, \\
\end{cases}\\
&\mathbf{V}_i=
\begin{cases}
\mathbf{B}_{i''}, & i\in[K],\\
\sum_{l=1}^n q_{i-K,l} \mathbf{B}_l, & K+1\leq i \leq N, \\
\end{cases}
\end{align} where $p_{i,j}, q_{i,j}$ are realizations of $P_{i,j}, Q_{i,j}$ which are absolutely continuous random variables with respect to the Lebesgue measure. $\mathbf{U}_i$ and  $\mathbf{V}_i$ are then transmitted to the $i$th worker node which is tasked with computing
$\mathbf{X}_i = \mathbf{U}_i^{\mathsf{T}} \mathbf{V}_i$.
We will refer to worker nodes $1,2,\ldots,K$ as \emph{systematic worker nodes} and
we will refer to worker nodes $K+1,\ldots,N$ as \emph{parity worker nodes}.

As in the case of non-systematic encoding, we focus on the recovery of the $(1,1)$th entry of $\mathbf{A}_j^{\mathsf{T}} \mathbf{B}_l$, namely $[\mathbf{A}_j^{\mathsf{T}} \mathbf{B}_l]_{1,1}$.
The same idea extends to the recovery of other indices as well.
Let $y_i = [X_i]_{1,1}$ denote the $(1,1)$th entry in the matrix product computed by the $i$th non-straggler worker node and let $z_{j,l} = [\mathbf{A}_j^{\mathsf{T}} \mathbf{B}_l]_{1,1}$.
For $j \in [mn]$, let $j' = \lceil j/n \rceil$ and $j'' = ((j-1) \mod n) + 1$ and let
$w_j = z_{j',j''}$.
The computed values $y_i$'s are related to the unknown values $w_j$'s according to

\begin{equation}
\label{eqn:sysmatrix1}\hspace{-0.125cm}
    \begin{bmatrix}
    y_1 \\
    y_2 \\
    \vdots\\
    y_K \\
    y_{K+1} \\
    \vdots \\
    y_N \\
    \end{bmatrix}
    =
    \begin{bmatrix}
    1 & 0 & \ldots & 0 & \ldots & 0 \\
    0 & 1 & \ldots & 0 & \ldots & 0 \\
    & & \vdots & & \\
    0 & 0 & \ldots & 0 & \ldots & 1 \\
    p_{1,1}q_{1,1} & p_{1,1}q_{1,2} & \ldots & p_{1,1}q_{1,n} & \ldots & p_{1m}q_{1,n}\\
    p_{2,1}q_{2,1} & p_{2,1}q_{2,2} & \ldots & p_{2,1}q_{2,n} & \ldots & p_{2m}q_{2,n}\\
    & & \vdots & & \\
    p_{S,1}q_{S,1} & p_{S,1}q_{S,2} & \ldots & p_{S,1}q_{S,n} & \ldots & p_{S,m}q_{S,n}\\
    \end{bmatrix}
    \begin{bmatrix}
    w_1 \\
    w_{2} \\
    \vdots\\
    w_{j} \\
    \vdots \\
    w_{K} \\
    \end{bmatrix},
\end{equation}
where $S=N-K$.

The generator matrix in \eqref{eqn:sysmatrix1} can be written as $[\mathbf{I}_{K \times K} \ \mathbf{F}^{\sf T}]^{\sf T}$ where $\mathbf{F}$ is an $N-K \times K$ matrix whose $i$th row is given by $\underline{p}_i \odot \underline{q}_i$, i.e.,

\begin{equation}
\label{eqn:sysparitymatrix}\hspace{-0.25cm}
\mathbf{F} =
    \begin{bmatrix}
    p_{1,1}q_{1,1} & p_{1,1}q_{1,2} & \ldots & p_{1,1}q_{1,n} & \ldots & p_{1m}q_{1,n}\\
    p_{2,1}q_{2,1} & p_{2,1}q_{2,2} & \ldots & p_{2,1}q_{2,n} & \ldots & p_{2m}q_{2,n}\\
    & & \vdots & & \\
    p_{S,1}q_{S,1} & p_{S,1}q_{S,2} & \ldots & p_{S,1}q_{S,n} & \ldots & p_{S,m}q_{S,n}\\
    \end{bmatrix}.
\end{equation}

\begin{definition}
\label{def:sysensemble}
The ensemble of $N \times K$ generator matrices obtained by choosing the generator matrix $\mb{G}$ as in \eqref{eqn:sysmatrix1} where $P_{i,j}, Q_{i,j} \sim f$ will be referred to as the \emph{systematic random Khatri-Rao-product generator matrix ensemble} $\mathcal{G}_{sys}(N,K,f)$.
\end{definition}

\subsection{Decoding}
\label{sec:sysdecoding}
We consider the case when there are $S_1$ stragglers among the systematic worker nodes and $S_2 = S-S_1$ stragglers among the parity worker nodes.
Without loss of generality we assume that the stragglers are the worker nodes $1,2,\ldots,S_1$ and
$K+S_1+1, \ldots, N$.
This implies that the master nodes obtains $y_{S_1+1}, \ldots, y_K$ and since the encoding is systematic, the master node can trivially recover $w_{S_1+1}, \ldots, w_{K}$ by setting $w_i = y_i$ for $i=S_1+1, \ldots, S_K$.
We can recover $w_1, \ldots, w_{S_1}$ from $y_{K+1}, \ldots, y_{K+S_1}$ as follows. Notice that
\begin{align}
\label{eqn:sysmatrix2}
    &\begin{bmatrix}
    y_{K+1} \\
    y_{K+2} \\
    \vdots \\
    y_{K+S_1} \\
    \end{bmatrix}
    =
    [\mathbf{F}]_{K+1:K+S_1,1:S_1}
    \begin{bmatrix}
    w_1 \\
    w_{2} \\
    \vdots\\
    w_{S_1} \\
    \end{bmatrix}  \\
    &\qquad \qquad\quad +
    [\mathbf{F}]_{K+1:K+S_1,S_1+1:K}
    \begin{bmatrix}
    w_{S_1+1} \\
    w_{S_1+2} \\
    \vdots\\
    w_{K} \\
    \end{bmatrix},
    \end{align} which in turn implies that
    \begin{align}\label{eqn:sysmatrix3}
    &\underbrace{
    \begin{bmatrix}
    y_{K+1} \\
    y_{K+2} \\
    \vdots \\
    y_{K+S_1} \\
    \end{bmatrix}
    -
    [\mathbf{F}]_{K+1:K+S_1,S_1+1:K}
    \begin{bmatrix}
    w_{S_1+1} \\
    w_{S_1+2} \\
    \vdots\\
    w_{K} \\
    \end{bmatrix}}_{\underline{y}} \\
     &\qquad\qquad=
    \underbrace{[\mathbf{F}]_{K+1:K+S_1,1:S_1}}_{\mathbf{G}_{sys}}
    \underbrace{
    \begin{bmatrix}
    w_1 \\
    w_{2} \\
    \vdots\\
    w_{S_1} \\
    \end{bmatrix}}_{\underline{w}}, \nonumber
\end{align} or more succinctly,
\begin{equation}\label{eqn:sysmatrix4}
  \underline{y} = \mathbf{G}_{sys} \ \underline{w}.
\end{equation} We can obtain an estimate of $\underline{w}$, namely $\underline{\hat{w}}$, as
\begin{equation}
\label{eqn:sysinverse}
    \underline{\hat{w}} = \mathbf{G}_{sys}^{-1} \ \underline{y}.
\end{equation}

Note that in the above description it is assumed that the stragglers were worker nodes $1,2,\ldots,S$ and $K+S_1+1, \ldots, N$.
However, the same ideas can be used for arbitrary sets of stragglers.
The following example will clarify this.
\begin{example}
Consider an example with $m=2,n=3, K=6$ with $N=10$ worker nodes.
Let the straggler nodes be the worker nodes 2,4,5, and 8.
In this case, we first recover $w_1,w_3,w_6$ by setting $w_1 = y_1, w_3 = y_3$ and $w_6 = y_6$.
Then, we recover $w_2, w_4,$ and $w_5$ from $w_1,w_3,w_6,y_7,y_9,y_{10}$ using
\begin{equation}
    \begin{bmatrix}
    y_{7} \\
    y_{9} \\
    y_{10} \\
    \end{bmatrix}
    -
    [\mathbf{F}]_{\{7,8,9\},\{1,3,6\}}
    \begin{bmatrix}
    w_{1} \\
    w_{3} \\
    w_{6} \\
    \end{bmatrix}
     =
    [\mathbf{F}]_{\{7,8,9\},\{2,4,5\}}
    \begin{bmatrix}
    w_{2} \\
    w_{4} \\
    w_{5} \\
    \end{bmatrix}.
\end{equation}
\end{example}

We show that if a generator matrix is chosen at random from the systematic RKRP ensemble $\mathcal{G}_{sys}(N,K,f)$, the encoding scheme in \eqref{eqn:sysmatrix1} result in an MDS code with probability 1.

\begin{theorem}
Systematic RKRP codes are MDS codes with probability 1.
\end{theorem}

\begin{proof}
To prove the theorem, we need to prove that
$\mb{G}_{sys}$ is full rank with probability 1.
The proof follows along the same lines as the proof of Theorem~\ref{th:1}.
\end{proof}

\section{Review of OrthoPoly codes}\label{sec:OrthoPoly}
In this section, we will briefly review OrthoPoly codes for the sake of completeness. Details can be found in~\cite{fahim2019}. The encoding scheme consists of the master node dividing the matrices $\mb{A}$ and $\mb{B}$ as in Section~\ref{sm&p} and subsequently computing \[\mb{u}_i^{\sf T}=\sum\limits_{j=1}^{m-1} T_j(x_i) \mb{A}_j^{\sf T},\quad  \mb{v}_i=\sum_{j=1}^{n-1} T_{jm}(x_i) \mb{B}_j\] where $T_r(x) = \cos(r(\cos^{-1}(x)))$ and $x_i = \cos(\frac{(2i-1)\pi}{2N})$, and sending $\mb{u}_i^{\sf T}$ and $\mb{v}_i$ to the $i$th worker node.
The $i$th worker node computes $\mb{u}_i^{\sf T} \mb{v}_i$ and sends the result back the master node for decoding.
Let us assume that worker nodes $i=1,\ldots,K$ ($K$ is defined as $K \triangleq mn$) return their outputs.
As before, we focus on the recovery of $[\mb{A}_i^{\sf T} \mb{B}_j]_{1,1}$.
If $y_i = [\mb{u}_i^{\sf T} \mb{v}_i]_{1,1}$, then
\begin{equation}
\label{eqn:OrthoPolyencoding}
    \underbrace{
    \begin{bmatrix}
    y_1 \\
    y_2 \\
    \vdots\\
    y_i \\
    \vdots \\
    y_K \\
    \end{bmatrix}}_{\underline{y}}
    =
    \underbrace{
    \begin{bmatrix}
    T_0(x_1)&\cdots&T_{K-1}(x_1) \\
    \vdots& \ddots &\vdots \\
    T_{0}(x_K) & \cdots & T_{K-1}(x_K)
    \end{bmatrix}
    }_{\mathbf{G}_{O}}
    \mb{H}
    \underbrace{
    \begin{bmatrix}
    w_1 \\
    w_{2} \\
    \vdots\\
    w_{j} \\
    \vdots \\
    w_{K} \\
    \end{bmatrix}}_{\underline{w}}
\end{equation} where $\mb{H}$ is a $K\times K$ matrix such that
\[
\mb{H}_{(r,(i-1)+(j-1)m+1)}=
\begin{cases}
1, &  r=(i-1)+(j-1)m+1,\\ & i=1, j \in [n]\\
\frac{1}{2}, &  r=(i-1)+(j-1)m +1, \\ & i\neq1, i\in [m], j \in [n]\\
\frac{1}{2}, &   r=|(i-1)-(j-1)m|+1,  \\ & i\neq1, i\in [m], j \in [n]\\
0, & \text{otherwise.}
\end{cases}
\]
An estimate of $\underline{w}$ ( $\underline{w}_{i}=[\mb{A}_j \mb{B}_l]_{1,1}$ such that $i=r+lm +1$, for $0\leq r \leq m-1$ and $0 \leq l \leq n-1$) is then obtained according to
\begin{equation}
    \label{eqn:OrthoPolydecoding}
    \hat{\underline{w}} = \mathbf{H}^{-1} \mathbf{G}_O^{-1} \underline{y}.
\end{equation}

\section{Decoding Complexities}
In this section, we briefly discuss the decoding complexity of systematic RKRP codes and OrthoPoly codes. Decoding of systematic RKRP codes involves two steps. It involves inversion of the $S_1 \times S_1$ matrix
$\mathbf{G}_{sys}$ in \eqref{eqn:sysmatrix4} whose complexity is $O(S_1^3)$. To retrieve every entry of $[\mathbf{A}_j^{\sf T} \mathbf{B}_l]$, we need to multiply $\mathbf{G}_{sys}^{-1}$ and $\underline{y}$ which requires $O(S_1^2)$ operations. This step needs to be repeated for
each of the $\frac{N_1N_3}{mn}$ entries of $[\mathbf{A}_j^{\sf T} \mathbf{B}_l]$ and hence,
the overall decoding complexity is $O(S_1^3 + S_1^2 \frac{N_1 N_3}{mn})$.

OrthoPoly codes cannot be easily implemented in systematic form because of the multiplication by $\mathbf{H}$. Hence, the decoding complexity of OrthoPoly codes involves inverting a $K \times K$ matrix followed by $\frac{N_1N_3}{mn}$ multiplication of a $K \times K$ matrix and a $K \times 1$ vector. The overall complexity is hence $O(K^3 + K^2 \frac{N_1 N_3}{mn})$.

Since $S_1 \leq K$, the average decoding complexity for systematic RKRP codes is lower than that of OrthoPoly codes and the worst-case complexities (when $S_1 = K$) are identical.

\section{Simulation results}
We now present simulation results to demonstrate the superior numerical stability of RKRP codes.
We performed Monte Carlo simulations of the encoding and decoding process by choosing
the entries of $\mb{A}$ and $\mb{B}$ to be realizations of i.i.d Gaussian random variables with zero mean and unit variance.
For the presented results, we have considered the recovery of the $(1,1)$th entry of $\mb{A}^T_j \mb{B}_l$.
The corresponding vector $\underline{w}$ is then a realization of the vector-valued random variable $\underline{W}$.
Then, we computed $\underline{y}$ using \eqref{eqn:mainequation},  \eqref{eqn:sysmatrix1}, and \eqref{eqn:OrthoPolyencoding} for non-systematic RKRP codes, systematic RKRP codes, and OrthoPoly codes, respectively.
We randomly chose a subset of $N-K$ worker nodes to be stragglers.
Then, we computed $\underline{\hat{w}}$ using \eqref{eqn:nonsysdecoding}, \eqref{eqn:sysinverse}, \eqref{eqn:OrthoPolydecoding} for non-systematic RKRP codes, systematic RKRP codes, and OrthoPoly codes, respectively. For each of these codes, we define the average relative error to be
\[
\eta_{\mathrm{ave}} : = \mathbb{E} \left[\frac{||\underline{W}-\underline{\hat{W}}||_2}{||\underline{W}||_2}\right]
\]
and we estimate $\eta_{\mathrm{ave}}$ from Monte Carlo simulations.

\subsection{MDS property}
Firstly, in several million simulations, we never observed any instance where the generator matrix $\mathbf{G}$ in \eqref{eqn:mainequation} or $\mathbf{G}_{sys}$ in \eqref{eqn:sysmatrix3} was singular, which provides empirical evidence to our claim that non-systematic and systematic RKRP codes are MDS codes with probability 1.

\subsection{Average relative error}
In Fig.~\ref{fig:averageerrorvsK}, we plot the average relative error as a function of $N$ when the total number of worker nodes is set to be $N=\myceil{K/(1-\alpha)}$ or $K= \myfloor{N(1-\alpha)}$.
This model is meaningful when we consider practical scenarios where each worker node fails with a fixed probability.
In the plots in Fig.~\ref{fig:averageerrorvsK}, $\alpha$ is fixed and $K$ and $N$ are varied.
The results are shown for $\alpha = 0.1$ and for OrthoPoly codes, non-systematic RKRP codes, and systematic RKRP codes.
It can be seen that the average relative error is several orders of magnitude lower for RKRP codes when $N$ is about 100.

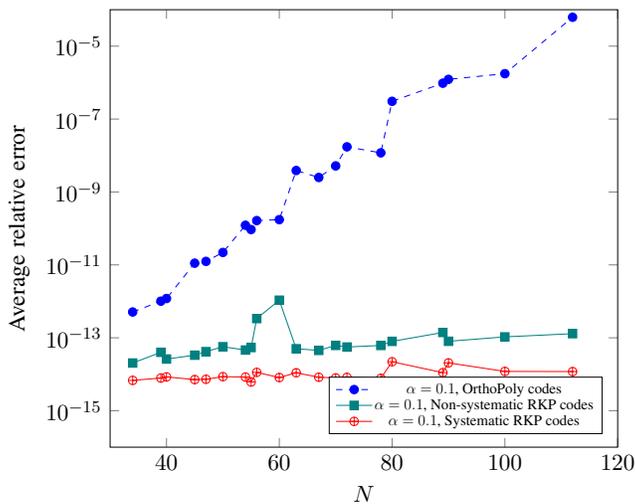
\begin{figure}[h]
\begin{center}
\begin{tikzpicture}[scale=0.8]
\begin{axis}[
scale only axis,
xmin=30, xmax=120, ymin=1e-16, ymax=1e-4, ymode=log,yminorticks=true,
xlabel=$N$, ylabel= Average relative error,legend pos=south east,legend style={nodes={scale=0.5, transform shape}}]

\addplot [color=blue,dashed,mark=*,mark options={solid}]
table[row sep=crcr]{%
34.0 5.103454736425935e-13\\
39.0 1.0105839931875792e-12\\
40.0 1.1950005358856895e-12\\
45.0 1.1136037419570017e-11\\
47.0 1.2480079463064443e-11\\
50.0 2.1938197023992022e-11\\
54.0 1.2206563421612598e-10\\
55.0 9.32145871881467e-11\\
56.0 1.6420278781541255e-10\\
60.0 1.7424949325257787e-10\\
63.0 3.884632503467416e-09\\
67.0 2.4999930951557986e-09\\
70.0 5.197028976938642e-09\\
72.0 1.731575102483833e-08\\
78.0 1.185243969625167e-08\\
80.0 3.0717108938515027e-07\\
89.0 9.629903986886668e-07\\
90.0 1.2283410011641096e-06\\
100.0 1.7566534775152742e-06\\
112.0 6.204123316249028e-05\\
};
\addlegendentry{\Large{$\alpha=0.1$, OrthoPoly codes}}
\addplot [color=teal,solid,mark=square*,mark options={solid}]
table[row sep=crcr]{%
34.0 2.034331651402009e-14\\
39.0 3.996795985875377e-14\\
40.0 2.6155781029681884e-14\\
45.0 3.334232859382861e-14\\
47.0 4.173529172584836e-14\\
50.0 5.696572674116812e-14\\
54.0 4.615855562893162e-14\\
55.0 5.405868449434661e-14\\
56.0 3.366922481304688e-13\\
60.0 1.074276547400439e-12\\
63.0 5.009868179163565e-14\\
67.0 4.5164474146039935e-14\\
70.0 6.192583652058885e-14\\
72.0 5.586775263356438e-14\\
78.0 6.171339666413338e-14\\
80.0 7.98421450617148e-14\\
89.0 1.3995028426677413e-13\\
90.0 8.024072147743728e-14\\
100.0 1.0598751560396825e-13\\
112.0 1.2979604284338983e-13\\
};
\addlegendentry{\Large{$\alpha=0.1$, Non-systematic RKP codes}}

\addplot [color=red,solid,mark=oplus,mark options={solid}]
table[row sep=crcr]{%
34 6.78586524823597e-15\\
39 7.877686121234068e-15\\
40 8.436583663483217e-15\\
45 7.140386709349934e-15\\
47 7.27955078157888e-15\\
50 8.588979091619948e-15\\
54 8.369709918013964e-15\\
55 6.134754123493365e-15\\
56 1.1257919466258076e-14\\
60 8.130254765262981e-15\\
63 1.0979158737891854e-14\\
67 8.278075236931348e-15\\
70 7.925906839573012e-15\\
72 8.229616850470538e-15\\
78 7.73192555982936e-15\\
80 2.2042799283907207e-14\\
89 1.1074083558778613e-14\\
90 2.0528029545865645e-14\\
100 1.1967144194625705e-14\\
112 1.1838635145219738e-14\\
};
\addlegendentry{\Large{$\alpha=0.1$, Systematic RKP codes}}

\end{axis}
\end{tikzpicture}
\end{center}
\caption{Plot of average relative error as a function of $N$ for a fixed $\alpha$; $N=\myceil{K/(1-\alpha)}$ } 
\label{fig:averageerrorvsK}
\end{figure}

In Figure~\ref{fig:averageerrorvsalpha}, we plot the average relative error versus $\alpha=\frac{N-K}{N}$ for a fixed $K$. Again, it can be seen that the proposed RKRP codes are very robust to numerical precision errors and substantially outperform OrthoPoly codes.
It should also be noted that the average relative error remains largely independent of $\alpha$ for RKRP codes whereas they grow rapidly with $\alpha$ for OrthoPoly codes.

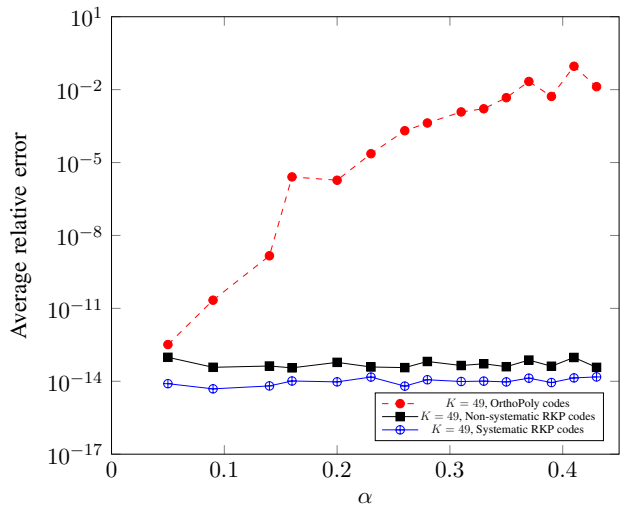
\begin{figure}[h]
\begin{center}
\begin{tikzpicture}[scale=0.8]
\begin{axis}[
scale only axis,
xmin=0, xmax=0.45, ymin=1e-17, ymax=1e+1, ymode=log,yminorticks=true,
xlabel= $\alpha$, ylabel= Average relative error ,legend pos=south east,legend style={nodes={scale=0.4, transform shape}}]

\addplot[color=red,dashed,mark=*,mark options={solid}]
table[row sep=crcr]{%
0.05 3.2403679482311323e-13\\
0.09 2.168526548520895e-11\\
0.14 1.4590511198845883e-09\\
0.16 2.5654761622466327e-06\\
0.20 1.8721826577801001e-06\\
0.23 2.3058298726250162e-05\\
0.26 0.00020400719099178546\\
0.28 0.0004239365417797338\\
0.31 0.0012103776590708915\\
0.33 0.0016291880832280787\\
0.35 0.004659355553004825\\
0.37 0.021655174667095083\\
0.39 0.005252098663182652\\
0.41 0.09168624735976076\\
0.43 0.013252551559512483\\
};
\addlegendentry{\Large{$K=49$, OrthoPoly codes}}
\addplot [color=black,solid,mark=square*,mark options={solid}]
table[row sep=crcr]{%
0.05 9.687014974033349e-14\\
0.09 3.757026106249502e-14\\
0.14 4.2447549898647975e-14\\
0.16 3.569059984966105e-14\\
0.20 6.001696493893898e-14\\
0.23 3.87050039335863e-14\\
0.26 3.6494681010075174e-14\\
0.28 6.526078050308505e-14\\
0.31 4.452487349537499e-14\\
0.33 5.2530635823973e-14\\
0.35 3.9766090519838594e-14\\
0.37 7.447579182968406e-14\\
0.39 4.1533903481105994e-14\\
0.41 9.542700579520118e-14\\
0.43 3.73693958524399e-14\\
};
\addlegendentry{\Large{$K=49$, Non-systematic RKP codes}}

\addplot [color=blue,solid,mark=oplus,mark options={solid}]
table[row sep=crcr]{%
0.05 7.962422360410414e-15\\
0.09 4.878298286498388e-15\\
0.14 6.438352583996073e-15\\
0.16 1.0263255386761733e-14\\
0.20 9.445206441243453e-15\\
0.23 1.4915571889782106e-14\\
0.26 6.321839290129278e-15\\
0.28 1.1538064635208609e-14\\
0.31 9.783359085435483e-15\\
0.33 1.0155074278683003e-14\\
0.35 9.41733020401526e-15\\
0.37 1.3437868940550588e-14\\
0.39 8.867724328873974e-15\\
0.41 1.3757915616201344e-14\\
0.43 1.5110463946165125e-14\\
};
\addlegendentry{\Large{$K=49$, Systematic RKP codes}}

\end{axis}
\end{tikzpicture}
\end{center}
\caption{Plot of average relative error versus fraction of straggler nodes ($\alpha$) for $K=49$; $\alpha=\frac{N-K}{N}$} 
\label{fig:averageerrorvsalpha}
\end{figure}

In Figure~\ref{fig:averageerrorvsS}, we plot the average relative error versus the number of straggler nodes $S$ for a fixed $K$ when $N = K + S$.
It can be seen that the proposed RKRP codes provide excellent robustness even as the number of stragglers increases.

\begin{figure}[h]
\begin{center}
\begin{tikzpicture}[scale=0.8]
\begin{axis}[
scale only axis,
xmin=0, xmax=28, ymin=1e-17, ymax=1e+1, ymode=log,yminorticks=true,
xlabel= Number of stragglers $S$ , ylabel= Average relative error ,legend pos=south east,legend style={nodes={scale=0.5, transform shape}}]
\addplot[color=red,dashed,mark=*,mark options={solid}]
table[row sep=crcr]{%
0 2.1972896322847754e-15\\
1 4.670500112342729e-15\\
2 3.5838124604631344e-14\\
3 3.3549615433598013e-13\\
4 3.1547465190495988e-12\\
5 1.8675423457975537e-11\\
6 4.2290706416321664e-10\\
7 5.371106439387565e-10\\
8 8.830348060046242e-09\\
9 1.2415987938191685e-08\\
10 5.4673991063796613e-08\\
11 2.4080757455940135e-07\\
12 1.5753990089404418e-06\\
13 2.812510527290255e-05\\
14 9.135844691831619e-06\\
15 7.440879992221428e-06\\
16 1.0679765470316333e-05\\
17 3.571445857162388e-05\\
18 0.00022725126829409916\\
19 0.000256306256079855\\
20 0.0002161793616319633\\
21 0.00040518743376102135\\
22 0.0002820353747702039\\
23 0.0012750089638207728\\
24 0.0017895920324810376\\
25 0.0063289549505900575\\
26 0.0060764365534497445\\
};
\addlegendentry{\Large{$K=49$, OrthoPoly codes}}

\addplot [color=black,solid,mark=square*,mark options={solid}]
table[row sep=crcr]{%
0 5.325408459627166e-14\\
1 3.316627337819264e-14\\
2 5.905620456580014e-14\\
3 1.4809624387091834e-13\\
4 5.455996967666671e-13\\
5 4.021229829302564e-14\\
6 4.050014350117992e-14\\
7 2.8195031970416047e-14\\
8 3.6662767216853286e-14\\
9 2.9767967443457914e-14\\
10 5.95002955007048e-14\\
11 2.8438926897403673e-14\\
12 4.1667435274129204e-14\\
13 4.769134137170719e-14\\
14 4.296110849488044e-14\\
15 3.044576519236672e-14\\
16 4.2166693641631226e-14\\
17 2.682788126504167e-14\\
18 4.504401883439918e-14\\
19 3.4710085894538557e-14\\
20 2.8383399563490676e-14\\
21 1.496582313655205e-13\\
22 4.102419323282304e-14\\
23 2.8248821016811223e-14\\
24 2.594864987861034e-14\\
25 3.2869108961003424e-14\\
26 4.30712336365571e-14\\
};
\addlegendentry{\Large{$K=49$, Non-systematic RKP codes}}

\addplot [color=blue,solid,mark=oplus,mark options={solid}]
table[row sep=crcr]{%
1 1.0030893030654464e-14\\
2 1.0939989040455154e-14\\
3 2.115775158514383e-14\\
4 6.841938598055812e-15\\
5 6.896289930137738e-15\\
6 9.520264914244478e-15\\
7 6.760761853091523e-15\\
8 7.191562766135915e-15\\
9 7.706845401887253e-15\\
10 2.080614298922878e-14\\
11 8.356230329016362e-15\\
12 9.746140672550001e-15\\
13 8.301308604923536e-15\\
14 1.2865752568167117e-14\\
15 9.545546774568183e-15\\
16 1.0169588703620527e-14\\
17 1.1099802026771581e-14\\
18 1.0163277424383382e-14\\
19 3.805223183684097e-14\\
20 9.316945746015226e-15\\
21 2.5056168059958586e-14\\
22 1.5872889038161014e-14\\
23 1.303282347017567e-14\\
24 9.951940809940353e-15\\
25 1.1662615679774973e-14\\
26 1.137852047272609e-14\\
};
\addlegendentry{\Large{$K=49$, Systematic RKP codes}}

\end{axis}
\end{tikzpicture}
\end{center}
\caption{Plot of average relative error versus number of stragglers ($S$) for $K=49$; $N=K+S$} 
\label{fig:averageerrorvsS}
\end{figure}
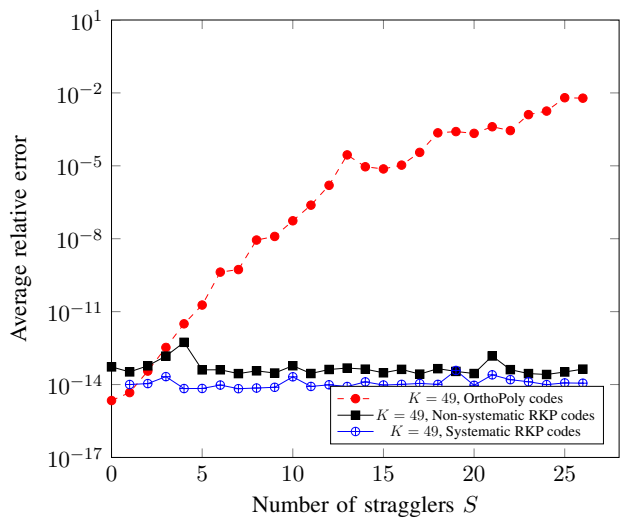

\subsection{Average log condition number}
The expected value of the logarithm of the condition number of a random matrix is a measure of loss in precision in computing the inverse of the determinant of the matrix, when the matrix is chosen from an underlying ensemble \cite{edelman1988eigenvalues}.
We computed the expected value of the logarithm of the condition number of
matrices from three ensembles. For non-systematic RKRP codes, we chose $\mathbf{G}$ from the $\mathcal{G}_{non-sys}(N,K,f)$ ensemble where $f$ is a Gaussian density with zero mean and unit variance.  For systematic RKRP codes, we chose $\mathbf{G}_{sys}$ from the $\mathcal{G}_{sys}(N,K,f)$ ensemble, and for OrthoPoly codes, we randomly chose $K \times K$ submatrices of $\mathbf{G}_{O}$ and multiplied the matrix by $\mathbf{H}$.
In Figure~\ref{fig:cnum}, we plot the average of the log of the condition number as a function of $\alpha$ for the three ensembles. We fix $K$ and let $N = K(1+\alpha)$.

It can be seen that the average of the log of the condition number is substantially lower for RKRP codes than for Orthopoly codes showing that the number of bits of precision lost is substantially lower for RKRP codes.

\begin{figure}[h]
\begin{center}
\begin{tikzpicture}[scale=0.8]
\begin{axis}[
scale only axis,
xmin=0, xmax=1, ymin=0, ymax=27,yminorticks=true,
xlabel= $\alpha$, ylabel= Average log(condition number) ,legend pos=south east,legend style={nodes={scale=0.4, transform shape}}]
\addplot[color=red,dashed,mark=*,mark options={solid}]
table[row sep=crcr]{%
0.0 2.547445130680347\\
0.057692307692307696 4.689861631516333\\
0.10909090909090909 6.710496193752571\\
0.15517241379310345 8.29151377087382\\
0.20967741935483872 9.964581071503398\\
0.25757575757575757 11.29264479267757\\
0.3 12.433174686869677\\
0.35526315789473684 13.76759514620506\\
0.4024390243902439 14.877592891423127\\
0.45555555555555555 16.07697895467572\\
0.5 17.059707664191183\\
0.5504587155963303 18.04892484630569\\
0.6016260162601627 19.112683069893446\\
0.65 20.0333693412482\\
0.7012195121951219 21.060702878271638\\
0.75 21.945721742233424\\
0.8008130081300813 22.935366456687042\\
0.8501529051987767 23.84851394488121\\
0.90020366598778 24.828563621826017\\
0.9500509683995922 25.840812135251667\\
};
\addlegendentry{\Large{$K=49$, OrthoPoly codes}}

\addplot [color=black,solid,mark=square*,mark options={solid}]
table[row sep=crcr]{%
0.0 5.921140933205746\\
0.057692307692307696 5.9186164746858445\\
0.10909090909090909 5.924071515202911\\
0.15517241379310345 5.918383902230549\\
0.20967741935483872 5.922510143895487\\
0.25757575757575757 5.918815339067264\\
0.3 5.919946295437675\\
0.35526315789473684 5.921771231797895\\
0.4024390243902439 5.920147601080497\\
0.45555555555555555 5.9185454666807935\\
0.5 5.92049975142468\\
0.5504587155963303 5.9219090486127985\\
0.6016260162601627 5.924555149875705\\
0.65 5.926503108255872\\
0.7012195121951219 5.930012223374312\\
0.75 5.922250725907986\\
0.8008130081300813 5.92250290251669\\
0.8501529051987767 5.922127690076054\\
0.90020366598778 5.924002945491798\\
0.9500509683995922 5.927227901612172\\
};
\addlegendentry{\Large{$K=49$, Non-systematic RKP codes}}

\addplot [color=blue,solid,mark=oplus,mark options={solid}]
table[row sep=crcr]{%
0.0 0.0\\
0.057692307692307696 2.6392646664149018\\
0.10909090909090909 3.5603614617122283\\
0.15517241379310345 3.9644763370797778\\
0.20967741935483872 4.293480646810595\\
0.25757575757575757 4.50569490271184\\
0.3 4.669198624763689\\
0.35526315789473684 4.835108073759988\\
0.4024390243902439 4.974816977442566\\
0.45555555555555555 5.091975207542352\\
0.5 5.191853298872646\\
0.5504587155963303 5.287568082047964\\
0.6016260162601627 5.377710151649313\\
0.65 5.4574855213789935\\
0.7012195121951219 5.54809874836106\\
0.75 5.61968740186733\\
0.8008130081300813 5.6841662303633615\\
0.8501529051987767 5.746394422889976\\
0.90020366598778 5.802789960723121\\
0.9500509683995922 5.868604082685714\\
};
\addlegendentry{\Large{$K=49$, Systematic RKP codes}}

\end{axis}
\end{tikzpicture}
\end{center}
\caption{Plot of E[log(condition number)] of inverted matrix versus fraction of straggler nodes ($\alpha$) for $K=49$} 
\label{fig:cnum}
\end{figure}
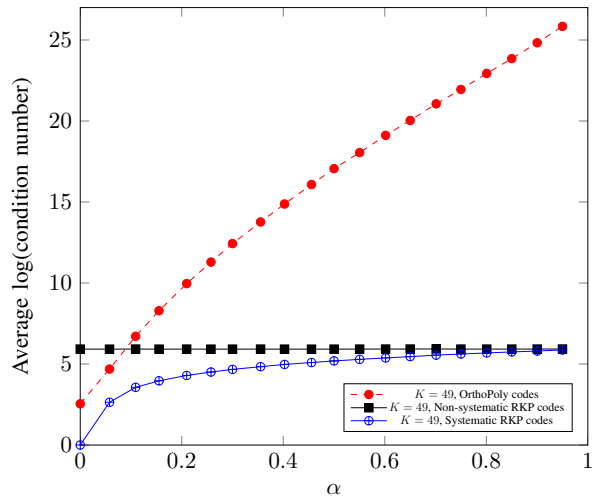

\section{Conclusion}
We proposed a new class of codes called random Khatri-Rao-product (RKRP) codes for which the generator matrix is the row-wise Khatri-Rao product of two random matrices. We proposed two random ensembles of generator matrices and corresponding codes called non-systematic RKRP codes and systematic RKRP codes. We showed that RKRP codes are maximum distance separable with probability 1 and that their decoding is substantially more numerically stable than Polynomial codes and OrthoPoly codes. The average decoding complexity of RKRP codes is lower than that of OrthoPoly codes.

\bibliographystyle{IEEEtran}
\bibliography{IEEEabrv,collab}
\end{document}